\documentclass[a4paper,reqno,twoside]{amsart}
\usepackage[left=3.5cm,right=3.5cm,top=3.5cm,bottom=3.5cm,footskip=1.5cm,headsep=1cm,bindingoffset=0.cm,]{geometry}
\usepackage{hyperref}
\usepackage{orcidlink}

\usepackage{cleveref}
\crefname{equation}{}{}
\crefname{figure}{Figure}{Figures}
\crefname{table}{Table}{Tables}
\crefname{section}{Section}{Sections}

\newtheorem{theorem}{Theorem}[section]

\newtheorem{corollary}[theorem]{Corollary}

\theoremstyle{definition}

\theoremstyle{remark}

\numberwithin{equation}{section}

\begin{document}

\title[Exceptional points and defective resonances]{Exceptional points and defective resonances in an acoustic scattering system with sound-hard obstacles}

\author[K. Matsushima]{Kei Matsushima\,\orcidlink{0000-0002-0352-8770}}
\address{\parbox{\linewidth}{Kei Matsushima\\
  Graduate School of Engineering, University of Tokyo, 2–11–16 Yayoi, Bunkyo–ku, Tokyo 113–8656, Japan, \href{http://orcid.org/0000-0002-0352-8770}{orcid.org/0000-0002-0352-8770}}.}
  \email{matsushima@mid.t.u-tokyo.ac.jp}

\author[T. Yamada]{Takayuki Yamada\,\orcidlink{0000-0002-5349-6690}}
\address{\parbox{\linewidth}{Takayuki Yamada\\
  Graduate School of Engineering, University of Tokyo, 2–11–16 Yayoi, Bunkyo–ku, Tokyo 113–8656, Japan, \href{http://orcid.org/0000-0002-5349-6690}{orcid.org/0000-0002-5349-6690}}.}

\begin{abstract}
This paper is concerned with non-Hermitian degeneracy and exceptional points associated with resonances in an acoustic scattering problem with sound-hard obstacles. The aim is to find non-Hermitian degenerate (defective) resonances using numerical methods. To this end, we characterize resonances of the scattering problem as eigenvalues of a holomorphic integral operator-valued function. This allows us to define defective resonances and associated exceptional points based on the geometric and algebraic multiplicities. Based on the theory on holomorphic Fredholm operator-valued functions, we show fractional-order sensitivity of defective resonances with respect to operator perturbation. This property is particularly important in physics and associated with intriguing phenomena, e.g., enhanced sensing and dissipation. A defective resonance is sought based on the perturbation analysis and Nystr\"om discretization of the boundary integral equation. Numerical evidence of the existence of a defective resonance is provided. The numerical results combined with theoretical analysis provide a new insight into novel concepts in non-Hermitian physics.
\end{abstract}

\maketitle

\section{Introduction}
Resonance characterizes a number of physical phenomena. For instance, it is well known that the sound of musical instruments originates from resonance of acoustic waves. Resonance is also the key to understanding the physics of \textit{metamaterials} \cite{engheta2006metamaterials}, which exhibit anomalous acoustic, elastic, electromagnetic, or quantum properties when subjected to external sources.

Resonance is often referred to as a phenomenon that occurs when an excitation frequency matches a natural frequency of a system, e.g., wave amplification and phase shift. A natural frequency is defined as the rate at which a system oscillates in time without any external sources. This is often characterized by an eigenvalue of a matrix or differential operator, e.g., Hamiltonian or Laplacian, acting on wave fields propagating in space.

A natural frequency can be complex in general. This is the case when a system involves energy loss and/or gain, which forces the time-harmonic oscillation to grow or decay exponentially in time. The growth/decay rate is measured from the imaginary part of a natural frequency. From the mathematical point of view, this energy conservation is associated with the Hermiticity or self-adjointness of the underlying matrix or operator. 

The breakdown of energy conservation induces not only the temporal growth and decay but also many interesting phenomena, which are intensively studied in the name of non-Hermitian physics \cite{ashida2020non-hermitian}. For instance, the non-Hermiticity of a matrix implies the existence of \textit{exceptional points} \cite{kato1966perturbation}. In non-Hermitian physics, exceptional points are referred to as a system's parameters for which two or more eigenvalues and corresponding eigenmodes coalesce simultaneously. In other words, an exceptional point is a singularity in a parameter space at which a matrix has a degenerate eigenvalue with the geometric multiplicity less than the algebraic one. While exceptional points are closely related to other anomalies in non-Hermitian systems, including non-Hermitian skin effects \cite{yao2018edge,ammari2024mathematical,ammari2024stability,matsushima2024non-bloch}, one of the most straightforward and important applications of exceptional points is the enhancement of acoustic and optical sensors based on the fractional-power eigenvalue splitting, as we shall see in \cref{s:mechanical}.

Although the non-Hermiticity is essential for such anomalies, some physical models exhibit similar phenomena originating from usual self-adjoint operators. For instance, acoustic wave scattering in unbounded media is described by the Helmholtz equation and thus characterized by the spectrum of the Laplace operator, which is self-adjoint on $L^2$. This, at first glance, suggests that we cannot observe the unique phenomena such as exceptional points and non-Hermitian skin effects in the scattering system. However, an interesting point is that generalized eigenvalues, or \textit{resonances} \cite{lax1990scattering,hislop1996introduction,dyatlov2019mathematical}, of the Laplace operator induces analogous non-Hermitian effects on scattered fields. Physically, these phenomena are caused from energy loss due to radiation and observed in some classical wave systems \cite{wiersig2011nonorthogonal,kullig2018exceptional,yi2018pair,abdrabou2019exceptional,bulgakov2021exceptional,gwak2021rayleigh,matsushima2023exceptional,deguchi2024observation}. Recent studies in physics have discussed the mechanism and application of such anomalies, e.g. \cite{heiss2012physics,hodaei2017enhanced,chen2017exceptional,rechtsman2017optical,miri2019exceptional}, especially in systems with parity--time symmetry. These non-Hermitian effects have been observed in multi-material models with high-contrast properties, where subwavelength high-quality resonances can be found. 

The aim of this study is to find such non-Hermitian phenomena in a scattering system with only sound-hard obstacles characterized by the homogeneous Neumann boundary condition. This model is simpler than the multi-material systems; however the identification of non-Hermitian phenomena becomes much more challenging as the trivial subwavelength resonances are no longer available. To this end, we formulate the scattering problem using a boundary integral equation and layer potentials. This allows us to characterize resonances in a straightforward manner, i.e., eigenvalues of a holomorphic operator-valued function \cite{taylor2023partial,misawa2017boundary,steinbach2017combined,spiridonov2020mathematical,Shestopalov2023resonance}. In view of this, we give a mathematical definition of the non-Hermitian degeneracy and exceptional points associated with a holomorphic Fredholm operator-valued function on Banach spaces. Our ultimate goal is to find non-Hermitian degenerate (or \textit{defective}) resonances and associated exceptional points numerically. For this purpose, we perform a perturbation analysis of resonances and show that defective resonances are characterized by its fractional-power sensitivity in response to an operator perturbation. 

The rest of this paper is organized as follows. \Cref{s:mechanical} is devoted to a brief introduction to the concept of non-Hermitian degenerate (defective) eigenvalues and exceptional points. In \cref{s:exterior-neumann}, we characterize a resonance in an exterior scattering problem as an eigenvalue of a holomorphic family of boundary integral operators. This motivates us to define defective degeneracy of eigenvalues associated with a holomorphic Fredholm operator-valued function. Numerical methods for the computation of resonances are summarized in \cref{s:numerical}. We demonstrate some numerical evidence of the existence of a defective resonance in \cref{s:result}. \Cref{s:conclusion} concludes the paper.

\section{Defective eigenvalue and exceptional point in a discrete mechanical system}\label{s:mechanical}
In this section, we give an introductory example of defective eigenvalues and exceptional points.

\begin{figure}
    \centering
    \includegraphics[scale=0.45]{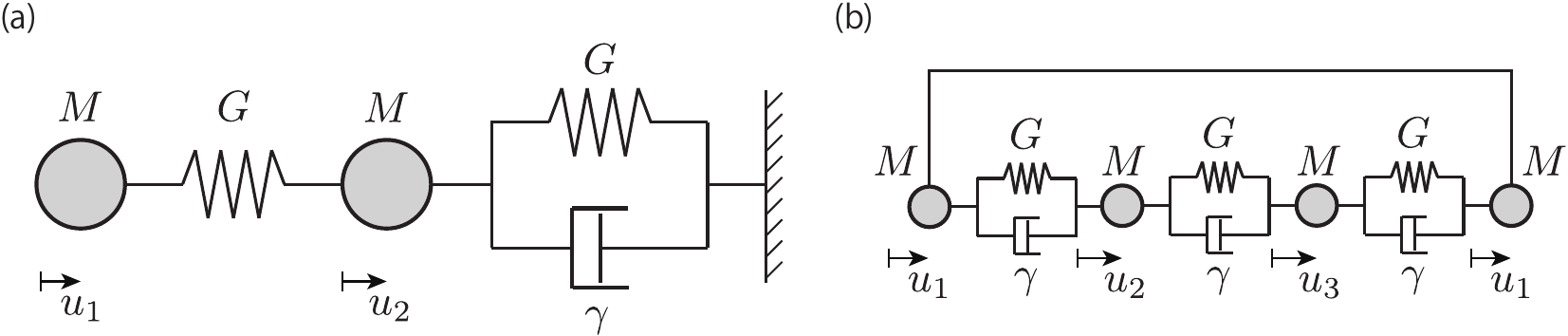}
    \caption{Discrete mechanical systems. The point masses, springs, and dampers are characterized by positive constants $M$, $G$, and $\gamma$, respectively.}
    \label{fig:discrete-1}
\end{figure}
Let us consider the simple mass-spring-damper model shown in \cref{fig:discrete-1} (a). The displacements $u_1$ and $u_2$ of the two point masses are governed by the following equation of motion:
\begin{align}
    \partial_t
    \begin{pmatrix}
        u_1(t)  \\ u_2(t) \\ M\dot u_1(t) \\ M\dot u_2(t)
    \end{pmatrix}
    = T
    \begin{pmatrix}
        u_1(t)  \\ u_2(t) \\ M\dot u_1(t) \\ M\dot u_2(t)
    \end{pmatrix}
    ,\quad 
    T
    :=
    \begin{bmatrix}
        0 & 0 & \frac{1}{M} & 0 
        \\
        0 & 0 & 0 & \frac{1}{M}
        \\
        - G &  G & 0 & 0
        \\
        G & -2G & 0 & -\frac{\gamma}{M}
    \end{bmatrix},
\label{eq:eom}
\end{align}
where $G$, $M$, and $\gamma$ are positive constants, and $\dot u_j:= \partial_t u_j$ denotes the velocity ($j=1,2$). A solution to the system of ordinary differential equations \cref{eq:eom} can be written using the matrix exponential $\mathrm{e}^{Tt}$ as
\begin{align*}
    \begin{pmatrix}
        u_1(t)  \\ u_2(t) \\ M\dot u_1(t) \\ M\dot u_2(t)
    \end{pmatrix}
    =
    \mathrm{e}^{Tt}
    \begin{pmatrix}
        u_1(0)  \\ u_2(0) \\ M\dot u_1(0) \\ M\dot u_2(0)
    \end{pmatrix}
    .
\end{align*}
Thus the time evolution is completely described by the eigenvalues and corresponding eigenvectors of the matrix $T$. Let $\lambda$ be an eigenvalue of $T$, i.e., 
\begin{align*}
    P_T(\lambda) := \mathrm{det}(\lambda I - T) = 0.
\end{align*}
The polynomial $P_T(\lambda)$ has common roots if and only if its discriminant vanishes, i.e., 
\begin{align*}
    0 = \mathrm{Res}(P_T,P_T^\prime) = \frac{G^3}{M^9} (\gamma^2 - 4GM)^2 (25GM - 4\gamma^2),
\end{align*}
where $\mathrm{Res}$ denotes the resultant. In view of this, we set $\gamma = 2\sqrt{GM}$. Then we have
\begin{align*}
    P_T(\lambda) = \left( \lambda - \lambda_1 \right)^2 \left( \lambda - \lambda_2 \right)^2,
    \quad
    \lambda_1 := \frac{-1-\mathrm i\sqrt{3}}{2}\omega_0,
    \quad
    \lambda_2 := \frac{-1+\mathrm i\sqrt{3}}{2}\omega_0
\end{align*}
with $\omega_0 := \sqrt{G/M}$. It is easy to see that their respective eigenspaces have dimension one with
\begin{align*}
    N(\lambda_1 I - {T}) =& \mathrm{span} \left\{ 
    \begin{pmatrix}
        \frac{1+\sqrt{3}\mathrm i}{2\sqrt{GM}},  \frac{-1+\sqrt{3}\mathrm i}{2\sqrt{GM}}, \frac{1-\sqrt{3}\mathrm i}{2},1
    \end{pmatrix}
    \right\},
\\
    N(\lambda_2 I - {T}) =& \mathrm{span} \left\{ 
    \begin{pmatrix}
        \frac{1-\sqrt{3}\mathrm i}{2\sqrt{GM}},  \frac{-1-\sqrt{3}\mathrm i}{2\sqrt{GM}}, \frac{1+\sqrt{3}\mathrm i}{2},1
    \end{pmatrix}
    \right\},
\end{align*}
where $N$ denotes the null space. In other words, the geometric multiplicities of $\lambda_1$ and $\lambda_2$ are one and thus less than their algebraic multiplicities. Thus, the matrix $T$ is not diagonalizable in this case. Instead, we have the following Jordan canonical form:
\begin{align*}
    T = U_T
    \begin{bmatrix}
        \lambda_1 & 1 & & 
        \\
         & \lambda_1 & &
        \\
         & & \lambda_2 & 1
         \\
         & & & \lambda_2
    \end{bmatrix}
    U_T^{-1},
    \quad
    U_T := 
    \begin{bmatrix}
        \frac{1+\sqrt{3}\mathrm i}{2\sqrt{GM}} & \frac{2+\sqrt{3}\mathrm i}{G} & \frac{1-\sqrt{3}\mathrm i}{2\sqrt{GM}} & \frac{2-\sqrt{3}\mathrm i}{G}
        \\
        \frac{-1+\sqrt{3}\mathrm i}{2\sqrt{GM}} & \frac{1+\sqrt{3}\mathrm i}{2G} & \frac{-1-\sqrt{3}\mathrm i}{2\sqrt{GM}} & \frac{1-\sqrt{3}\mathrm i}{2G}
        \\
        \frac{1-\sqrt{3}\mathrm i}{2} & \frac{1-\sqrt{3}\mathrm i}{\omega_0} & \frac{1+\sqrt{3}\mathrm i}{2} & \frac{1+\sqrt{3}\mathrm i}{\omega_0}
        \\
        1 & 0 & 1 & 0
    \end{bmatrix}
    .
\end{align*}
We say that an eigenvalue of a finite-dimensional square matrix is \textit{semi-simple} if its geometric and algebraic multiplicities coincide; otherwise it is called \textit{defective} or \textit{non-Hermitian degenerate}. It is obvious that any Hermitian matrix does not have defective (non-Hermitian degenerate) eigenvalues. 

Note that defective degeneracy should be distinguished from the usual semi-simple degeneracy, which often originates from a spatial symmetry in underlying physical models. To see this, let us consider the circular chain shown in \cref{fig:discrete-1} (b), which is described by the following equation of motion:
\begin{align*}
    \partial_t
    \begin{pmatrix}
        u_1(t) \\ u_2(t) \\ u_3(t) \\ M\dot u_1(t) \\ M\dot  u_2(t) \\ M\dot u_3(t)
    \end{pmatrix}
    =
    \begin{bmatrix}
        0 & 0 & 0 & \frac{1}{M} & 0 & 0
    \\
        0 & 0 & 0 & 0 & \frac{1}{M} & 0
    \\
        0 & 0 & 0 & 0 & 0 & \frac{1}{M} 
    \\
        -2G & G & G & -2\frac{\gamma}{M} & \frac{\gamma}{M} & \frac{\gamma}{M}
    \\
        G & -2G & G & \frac{\gamma}{M} & -2\frac{\gamma}{M} & \frac{\gamma}{M} 
    \\
        G & G & -2G & \frac{\gamma}{M} & \frac{\gamma}{M} & -2\frac{\gamma}{M}
    \end{bmatrix}
    \begin{pmatrix}
        u_1(t) \\ u_2(t) \\ u_3(t) \\ M\dot u_1(t) \\ M\dot  u_2(t) \\ M\dot u_3(t)
    \end{pmatrix}.
\end{align*}
The matrix on the right-hand side has three distinct eigenvalues
\begin{align*}
    \lambda_1 := 0,\quad \lambda_2 := \frac{-3\gamma -\mathrm i \sqrt{12MG - 9\gamma^2}}{2M}, \quad \lambda_3 := \frac{-3\gamma +\mathrm i \sqrt{12MG - 9\gamma^2}}{2M}.
\end{align*}
The zero eigenvalue $\lambda_1$ is defective and associated with rigid body motion with constant speed. The other eigenvalues $\lambda_2$ and $\lambda_3$ are both degenerate but have the same geometric and algebraic multiplicities of two if $12MG-9\gamma^2 \neq 0$, implying that they are semi-simple and not defective.

An important property of defective eigenvalues lies in its perturbation analysis. Let us consider again the four-by-four matrix $T$, defined in \cref{eq:eom}, and set $\gamma=2\sqrt{GM}$. We are interested in how a small perturbation on the matrix $T$ affects its eigenvalue distribution. For example, a straightforward calculation shows that the perturbed matrix $T_\varepsilon := T + \varepsilon B$ has four eigenvalues $(\lambda_1^+(\varepsilon),\lambda_1^-(\varepsilon),\lambda_2^+(\varepsilon),\lambda_2^-(\varepsilon))$ for each $\varepsilon\in\mathbb C$, given by
\begin{align*}
    \lambda^\pm_1(\varepsilon) &= \frac{\omega_0\left(-1\pm\sqrt{\varepsilon} - \mathrm i\sqrt{3-\varepsilon \pm 2\sqrt{\varepsilon}} \right)}{2} = \lambda_1 \pm \frac{\omega_0}{2}\left(1-\frac{\mathrm i}{\sqrt{3}} \right)\sqrt{\varepsilon} + O(\varepsilon),
    \\
    \lambda^\pm_2(\varepsilon) &= \frac{\omega_0\left(-1\pm\sqrt{\varepsilon} + \mathrm i\sqrt{3-\varepsilon \pm 2\sqrt{\varepsilon}} \right)}{2} = \lambda_2 \pm \frac{\omega_0}{2}\left(1+\frac{\mathrm i}{\sqrt{3}} \right)\sqrt{\varepsilon} + O(\varepsilon),
\end{align*}
where $B$ is the square matrix defined by
\begin{align*}
    B=
    \begin{bmatrix}
        0 & 0 & 0 & 1/M
    \\
        0 & 0 & 0 & 0
    \\
        0 & 0 & 0 & 0
    \\
        0 & 0 & 0 & 0
    \end{bmatrix}.
\end{align*}
This suggests that the original defective eigenvalues $\lambda_1$ and $\lambda_2$ split into the pairs $\{\lambda_1^+(\varepsilon),\lambda_1^-(\varepsilon)\}$ and $\{\lambda_2^+(\varepsilon),\lambda_2^-(\varepsilon)\}$, respectively, when $\varepsilon$ varies from zero. Note that such perturbation $B$ does not necessarily reflect the variation in the underlying physical model. We say that $\varepsilon_0\in\mathbb C$ is an \textit{exceptional point} of a family of square matrices $\{T_\varepsilon\}_{\varepsilon\in\mathbb C}$ if there exists a neighborhood $V$ of $\varepsilon_0$ and $m,n\in\mathbb N$ with $m>n$ such that
\begin{enumerate}
    \item the number of distinct eigenvalues of $T_{\varepsilon_0}$ is $n$ and
    \item the number of distinct eigenvalues of $T_{\varepsilon}$ is $m$ for all $\varepsilon\in V\setminus\{\varepsilon_0\}$.
\end{enumerate} 
In this example, the point $\varepsilon=0$ is an exceptional point of the family $\{ T+\varepsilon B\}_\varepsilon$.

These discussions on eigenvalues of (a family of) square matrices are based on analytic properties of the matrix-valued function $\lambda\mapsto \lambda I - T =:A(\lambda)$. We shall generalize these concepts, e.g., multiplicities, defective degeneracy, and exceptional points, to a class of operator-valued functions $A$ on infinite-dimensional Banach spaces.

\section{Resonances in exterior Neumann problem}\label{s:exterior-neumann}
Let $\Omega$ be an open and bounded subset of $\mathbb{R}^2$ with boundary $\partial\Omega$ of class $C^2$. We assume that the exterior $\mathbb R^2\setminus\overline\Omega$ is connected. For a given constant $k\in\mathbb{C}$ and function $g\in C(\partial\Omega)$, we consider the following exterior Neumann problem: find $u\in C^2(\mathbb R^2\setminus\overline\Omega) \cap C(\mathbb R^2\setminus\Omega)$ such that
\begin{align}
    \begin{cases}
    \displaystyle\varDelta u + k^2 u = 0 &\text{in }\mathbb{R}^2\setminus\overline\Omega, 
    \\
    \displaystyle\frac{\partial u}{\partial \nu} = g &\text{on }\partial\Omega, 
    \\
    \displaystyle\sqrt{r} \left( \frac{\partial u}{\partial r} - \mathrm ik u  \right) \to 0 &\text{uniformly as }r:=|x|\to+\infty,
    \end{cases}
    \label{eq:exterior-neumann}
\end{align}
where $\frac{\partial}{\partial\nu} := \nu\cdot\nabla$ denotes the normal derivative with unit normal vector $\nu$ outward to $\Omega$. Here we summarize well-known facts about the exterior Neumann problem. See \cite{colton2013integral,colton2013inverse,steinbach2017combined} for details. The exterior Neumann problem \cref{eq:exterior-neumann} is uniquely solvable for all $k\in\mathbb C$ with positive imaginary part, and the solution is given by
\begin{align}
    u(x) = \int_{\partial\Omega} \left( \frac{\partial \Phi}{\partial\nu(y)}(x,y;k) \varphi(y) - \Phi(x,y;k) g(y) \right)\mathrm ds(y) \quad\text{for all }x\in\mathbb R^2\setminus\overline\Omega, \label{eq:sol-formula}
\end{align}
where $\Phi$ is the fundamental solution of the two-dimensional Helmholtz equation, given by $\Phi(x,y;k)=\frac{\mathrm i}{4}H^{(1)}_0(k|x-y|)$. The density $\varphi\in C(\partial\Omega)$ is a unique solution to the following boundary integral equation:
\begin{align}
    (I - D_k)\varphi = S_k g, \label{eq:bie}
\end{align}
where $I$ is the identity operator, $S_k:C(\partial\Omega)\to C(\partial\Omega)$ and $D_k:C(\partial\Omega)\to C(\partial\Omega)$ are the boundary integral operators defined by
\begin{align*}
    (S_k \varphi)(x) &= 2\int_{\partial\Omega} \Phi(x,y;k)\varphi(y)\mathrm ds(y) ,
    \\
    (D_k \varphi)(x) &= 2\int_{\partial\Omega} \frac{\partial \Phi}{\partial\nu(y)}(x,y;k)\varphi(y)\mathrm ds(y) 
\end{align*}
for all $\varphi\in C(\partial\Omega)$ and $x\in\partial\Omega$. The linear operators $S_k$ and $D_k$ on $C(\partial\Omega)$ are compact. Furthermore, we observe that the mapping $k\mapsto S_k$ and $k\mapsto D_k$ from the upper half-plane onto $\mathcal L(C(\partial\Omega),C(\partial\Omega))$, the complex Banach space of all bounded linear operators on $C(\partial\Omega)$, is holomorphic. They have analytic continuation from the upper half-plane to the lower half-plane through the positive real axis. Using the continuation, we pose the same boundary integral equation \cref{eq:bie} for $k\in\mathbb C$ with negative imaginary part. We are interested in the case where the equation is not uniquely solvable, i.e., the values of $k\in\Lambda := \{ z\in\mathbb C : \mathrm{Im}[z]<0 \}$ for which $I-D_k$ is not invertible. By the compactness of $D_k$, this is equivalent that $I-D_k$ is not injective. We say that $k\in\Lambda$ is a \textit{(scattering) resonance} of the exterior Neumann problem if $N(I-D_k)$, the null space of $I-D_k$, is not trivial. For every element $\varphi\neq 0$ in $N(I-D_k)$ the function \cref{eq:sol-formula} solves the Helmholtz equation with the homogeneous Neumann boundary condition. This function on $\mathbb R^2\setminus\overline\Omega$ is called a \textit{resonant state} of the exterior Neumann problem.

\subsection{Review of the theory on holomorphic Fredholm operator-valued functions}
We have seen that resonances are characterized by the operator-valued function $k\mapsto I-D_k$ with a compact operator $D_k$. This motivates us to introduce the following theory on holomorphic Fredholm operator-valued functions. Further details can be found in, e.g., \cite{kozlov1999differential}.

Let $\Lambda$ be an open and connected subset of $\mathbb C$ and let $X$ and $Y$ be complex Banach spaces. An operator-valued function $A:\Lambda\to \mathcal L(X,Y)$ is called \textit{holomorphic} if it is differentiable in the norm at every point in $\Lambda$. We are interested in the subset $\sigma(A) := \{k\in\Lambda: A(k)\text{ is not invertible} \}$, called the \textit{spectrum} of $A$. In particular, $k\in\sigma(A)$ is called an \textit{eigenvalue} of $A$ if $A(k)$ is not injective, i.e., the \textit{eigenspace} $N(A(k))$ is not trivial. The dimension of an eigenspace is called the \textit{geometric multiplicity}. For each eigenvalue $k\in\sigma(A)$, a nonzero vector in $N(A(k))$ is called an \textit{eigenvector} of $A$ associated with $k$. 

In particular, we say that $A$ is a \textit{holomorphic Fredholm operator-valued function} if 
\begin{enumerate}
    \item $A$ is holomorphic,
    \item $A(k):X\to Y$ is a Fredholm operator for all $k\in\Lambda$, and
    \item there exists at least one $k\in\Lambda$ such that $A(k)$ is invertible.
\end{enumerate}
 In what follows, we always assume that $A$ is a holomorphic Fredholm operator-valued function. The analytic Fredholm theorem ensures that $\sigma(A)$ is a discrete subset of $\Lambda$ and the index of $A(k)$ is zero for all $k\in\Lambda$. The spectrum $\sigma(A)$ consists of all eigenvalues of $A$ and their respective eigenspace $N(A(k))$ is finite-dimensional, i.e., their geometric multiplicities are always finite.

The algebraic multiplicity will be defined via Jordan chains. An ordered collection $(\varphi_0,\varphi_1,\ldots,\varphi_{n-1})$ in $X$ is called a Jordan chain of $A$ associated with an eigenvalue $k\in\sigma(A)$ if $\varphi_0 \neq 0$ and 
\begin{align*}
    \sum_{l=0}^j \frac{1}{l!}A^{(l)}(k)\varphi_{j-l} = 0 \quad \text{for all }j=0,\ldots,n-1,
\end{align*}
where $A^{(l)}(k)$ denotes the $l$th-order derivative of $A$ at $k$. An element in $X$ is called a \textit{generalized eigenvector} of $A$ associated with $k\in\sigma(A)$ if it belongs to a Jordan chain at $k$. The algebraic multiplicity of an eigenvalue $k\in\sigma(A)$ is defined as the dimension of the linear hull spanned by all the generalized eigenvectors associated with $k$. 

An eigenvector is also a generalized eigenvector by definition; thus the algebraic multiplicity is greater than or equal to the geometric multiplicity. Furthermore, it can be seen that the algebraic multiplicity is always finite. An eigenvalue $k\in\sigma(A)$ is said to be \textit{semi-simple} if its geometric and algebraic multiplicities coincide; otherwise it is called \textit{defective}.

\subsection{Characterization of defective eigenvalues via perturbation analysis}\label{ss:characterization}
In terms of the above definitions, our aim is to numerically identify some defective eigenvalues of $k\mapsto I-D_k$. While it is easy to give concrete examples of semi-simple eigenvalues (resonances) of $k\mapsto I-D_k$ using some simple geometries, it is an open question whether a resonance of the exterior Neumann problem can be defective. 

Later we shall identify a defective resonance based on the following result \cite[Theorem 1]{stepin2006fredholm}, which is restated in our Banach space setting. 

\begin{theorem}\label{thm:main}
    Let $\Lambda$ be an open and connected subset of $\mathbb C$, let $X$ and $Y$ be complex Banach spaces, let $A:\Lambda\to\mathcal L(X,Y)$ be a holomorphic Fredholm operator-valued function, let $k_0\in\Lambda$ be an eigenvalue of $A$, and let $B\in\mathcal L(X,Y)$ be a nonzero operator.
    Then there exist open neighborhoods $V$ of $0\in\mathbb C$ and $\Lambda_0$ of $k_0$ and a matrix-valued function $F:V\times \Lambda_0\to \mathcal L(\mathbb C^m,\mathbb C^m)$, where $m$ is the geometric multiplicity of the eigenvalue $k_0$, such that 
    \begin{enumerate}
        \item for every $\varepsilon\in V$, $k\mapsto F(\varepsilon,k)$ is holomorphic on $\Lambda_0$,
        \item for every $k\in \Lambda_0$, $\varepsilon\mapsto F(\varepsilon,k)$ is holomorphic on $V$,
        \item for every $\varepsilon\in V$, $k\in\Lambda_0$ is an eigenvalue of $A+\varepsilon B$ if and only if $k$ is an eigenvalue of $k\mapsto F(\varepsilon,k)$, 
        \item the geometric multiplicities of the eigenvalue $k_0$ of $A$ and $k\mapsto F(0,k)$ coincide, and
        \item the algebraic multiplicities of the eigenvalue $k_0$ of $A$ and $k\mapsto F(0,k)$ coincide.
    \end{enumerate}
\end{theorem}
\begin{proof}
    Let $X^\prime$ be a complex Banach space with bounded and nondegenerate bilinear (or sesquilinear) form $\langle\cdot,\cdot\rangle:X^\prime\times X\to\mathbb C$ and let $u_1,\ldots,u_m$ be a basis of $N(A(k_0))$. Then there exist linearly independent vectors $\psi_1,\ldots,\psi_m\in X^\prime$ such that $\langle \psi_i,u_j\rangle = \delta_{ij}$ \cite[Lemma~4.14]{kress2014linear}. In addition, let $\varphi_1,\ldots,\varphi_m\in Y$ be linearly independent modulo $A(k_0)(X)$, where $A(k_0)(X)$ denotes the range of $A(k_0):X\to Y$. We define a compact operator $P:X\to Y$ by 
    \begin{align*}
        P\varphi = \sum_{j=1}^m \langle \psi_j,\varphi\rangle \varphi_j \quad\text{for all }\varphi\in X.
    \end{align*}
    
    We first show that there exist open neighborhoods $V$ of $0\in\mathbb C$ and $\Lambda_0$ of $k_0$ such that for every $\varepsilon\in V$ the operator $A(k)+\varepsilon B + P$ is invertible for all $k\in \Lambda_0$ and $k\mapsto A(k)+\varepsilon B$ is a holomorphic Fredholm operator-valued function on $\Lambda_0$. By the definition of $P$, the Fredholm operator $A(k_0)+P$ is invertible \cite[\S A.8]{kozlov1999differential}. Moreover, the Neumann series shows that $I+\varepsilon(A(k_0)+P)^{-1}B$ is invertible for all $\varepsilon\in \mathcal B(0;1/M)$, where $\mathcal B(c;r)$ denotes the open ball of radius $r>0$ centered at $c$ and $M:=\| (A(k_0)+P)^{-1} B\|> 0$, with estimate
    \begin{align}
        \| (I+\varepsilon(A(k_0)+P)^{-1}B)^{-1} \| \leq \frac{1}{1-M|\varepsilon|} \quad\text{for all } \varepsilon\in \mathcal B(0;1/M). \label{eq:est}
    \end{align}
    It is easy to check that $[I+\varepsilon(A(k_0)+P)^{-1} B]^{-1}(A(k_0)+P)^{-1}$ is the inverse of $A(k_0)+\varepsilon B+P$ for all $\varepsilon\in \mathcal B(0;1/M)$. Moreover, the estimate \cref{eq:est} shows that the mapping $\varepsilon\mapsto (A(k_0)+\varepsilon B+P)^{-1}$ is uniformly bounded on $\mathcal B(0;1/(2M))=:V_1$, i.e., there exists a constant $C>0$ such that $\| (A(k_0)+\varepsilon B+P)^{-1} \| \leq C$ for all $\varepsilon\in V_1$. 
    In combination with the continuity of $k\mapsto A(k)$, this implies that
    \begin{align*}
        \| (A(k)-A(k_0)) (A(k_0)+\varepsilon B+P)^{-1} \| < 1
    \end{align*}
    for all $\varepsilon\in V_1$ and $k$ in an open neighborhood $\Lambda_0$ of $k_0$. Thus the Neumann series again shows that $I + (A(k)-A(k_0)) (A(k_0)+\varepsilon B+P)^{-1}$ is invertible for all $\varepsilon\in V_1$ and $k\in \Lambda_0$. A straightforward calculation shows that 
    \begin{align*}
        (A(k_0)+\varepsilon B + P)^{-1} [I + (A(k)-A(k_0)) (A(k_0)+\varepsilon B+P)^{-1}]^{-1}
    \end{align*}
    is the inverse of $A(k)+\varepsilon B + P$.

    Since both the sets of invertible operators and Fredholm operators of index zero are open in $\mathcal L(X,Y)$ \cite[Theorem 4.1]{gohberg1990classes}, there exists an open neighborhood $V\subset V_1$ of $0\in\mathbb C$ such that for every $\varepsilon\in V$ the mapping $k\mapsto A(k)+\varepsilon B$ is a holomorphic Fredholm operator-valued function on $\Lambda_0$.

    The rest of the proof is analogous to \cite[Lemma~A.8.2]{kozlov1999differential}. Let $Q$ be the projection operator from $X$ onto $N(A(k_0))$, given by
    \begin{align*}
        Q\varphi = \sum_{j=1}^m \langle \psi_j,\varphi \rangle u_j \quad\text{for all }\varphi\in X.
    \end{align*}
    In addition, we define a bounded linear operator $L:X\to U\times\mathbb C^m$ with $U:=(I-Q)(X)$ by
    \begin{align*}
        L\varphi =
        \begin{pmatrix}
            (I-Q)\varphi \\ \langle\psi_1,\varphi \rangle \\ \vdots \\ \langle\psi_m,\varphi \rangle
        \end{pmatrix}
        \quad\text{for all }\varphi\in X.
    \end{align*}
    Then the operator $L$ has the inverse given by
    \begin{align*}
        L^{-1} \begin{pmatrix}
            w \\ \xi
        \end{pmatrix}
        = w + \sum_{j=1}^m \xi_j u_j \quad\text{for all }(w,\xi)\in U\times \mathbb C^m
    \end{align*}
    since $U\subset N(Q)$, $N(A(k_0))\subset N(I-Q)$, and $\langle \psi_j,w\rangle = 0$ for all $w\in U$ and $j=1,\ldots,m$.
    
    Let $\varepsilon\in V$ be fixed. It is easy to check that $I - (I-Q)(A(k_0)+\varepsilon B + P)^{-1} P$ is invertible with its inverse being $I + (I-Q)(A(k_0)+\varepsilon B + P)^{-1} P$ since $P(I-Q)=0$. Moreover, for every $k\in \Lambda_0$ we have
    \begin{align*}
        &\quad [I-Q(A(k)+\varepsilon B+P)^{-1}P] [I-(I-Q)(A(k)+\varepsilon B+P)^{-1}P] 
    \\
        &= I - (A(k)+\varepsilon B+P)^{-1}P.
    \end{align*}
    This implies that
    \begin{align*}
        &\quad (A(k)+\varepsilon B + P)L^{-1} L \left[I-Q(A(k)+\varepsilon B+P)^{-1}P \right]L^{-1}
    \\
        &\qquad\qquad \times L\left[I-(I-Q)(A(k)+\varepsilon B+P)^{-1}P\right] 
        = A(k)+\varepsilon B
    \end{align*}
    for all $k\in \Lambda_0$. Since $k\mapsto (A(k)+\varepsilon B + P)L^{-1}$ and $k\mapsto L[I-(I-Q)(A(k)+\varepsilon B+P)^{-1}P]$ are holomorphic on $\Lambda_0$ and the operators are invertible for all $k\in\Lambda_0$, the value $k\in\Lambda_0$ is an eigenvalue of the holomorphic Fredholm operator-valued function $A+\varepsilon B$ if and only if $k$ is an eigenvalue of $L[I-Q(A+\varepsilon B+P)^{-1}P]L^{-1}=:\tilde A_\varepsilon$. Moreover, the geometric and algebraic multiplicities of the eigenvalue of $A+\varepsilon B$ and $\tilde A_\varepsilon$ are the same, respectively \cite[Proposition~A.5.1]{kozlov1999differential}. 
    
    From the definition of $L$ with the properties $(I-Q)Q=0$ and $\langle \psi_j,Q\varphi\rangle = \langle \psi_j,\varphi\rangle$ for all $j=1,\ldots,m$, we have
    \begin{align*}
        \tilde A_\varepsilon(k)
        \begin{pmatrix}
            w \\ \xi
        \end{pmatrix}
        =
        \begin{pmatrix}
            w \\ F(\varepsilon,k)\xi
        \end{pmatrix}
         \quad\text{for all }(w,\xi)\in U\times \mathbb C^m \text{ and }k\in\Lambda_0,
    \end{align*}
    where $F(\varepsilon,k):\mathbb C^m\to\mathbb C^m$ is defined by
    \begin{align*}
        F(\varepsilon,k)\xi = \xi - 
        \begin{pmatrix}
            \displaystyle \sum_{j=1}^m \langle \psi_1,(A(k)+\varepsilon B+P)^{-1}\varphi_j\rangle \xi_j
            \\
            \vdots
            \\
            \displaystyle \sum_{j=1}^m \langle \psi_m,(A(k)+\varepsilon B+P)^{-1}\varphi_j\rangle \xi_j
        \end{pmatrix}
        \quad\text{for all }\xi\in\mathbb C^m.
    \end{align*}
    This implies that the invertibility of $\tilde A_\varepsilon(k)$ and $F(\varepsilon,k)$ is also the same, i.e., $k\in\Lambda_0$ is an eigenvalue of $\tilde A_\varepsilon$ if and only if it is an eigenvalue of $k\mapsto F(\varepsilon,k)$. Moreover, their geometric and algebraic multiplicities are the same. It is clear that $F:V\times\Lambda_0\to\mathbb \mathcal L(\mathbb C^m,\mathbb C^m)$ is holomorphic in each variable.
\end{proof}

In the special case of $m=1$, the matrix-valued function $F$ is just a complex-valued function in two complex variables whose order of vanishing at $(0,k_0)$ is the algebraic multiplicity $p$ of the eigenvalue $k_0$. This allows us to use the Weierstrass preparation theorem (e.g. \cite[Theorem~6.2.3]{lebl2024tasty}) as summarized below.
\begin{corollary}
    Under the assumptions of \Cref{thm:main} with $m=1$, there exist neighborhoods $V$ of $0\in\mathbb C$ and $\Lambda_0$ of $k_0$ and holomorphic functions $c_j:V\to\mathbb C$ ($j=0,\ldots,p-1$) such that
    \begin{enumerate}
        \item $c_0(0)=c_1(0)=\ldots=c_{p-1}(0)=0$,
        \item for every $\varepsilon\in V$, the polynomial
        \begin{align*}
            k\mapsto W(\varepsilon,k):= (k-k_0)^p + c_{p-1}(\varepsilon)(k-k_0)^{p-1} + \ldots + c_1(\varepsilon) (k-k_0) + c_0(\varepsilon),
        \end{align*}
        has $p$ zeros (counting multiplicity) in $\Lambda_0$, and
        \item for every $\varepsilon\in V$, $k\in \Lambda_0$ is an eigenvalue of $A+\varepsilon B$ if and only if $k\in \Lambda_0$ is a zero of the polynomial $k\mapsto W(\varepsilon,k)$, 
    \end{enumerate}
    where $p$ is the algebraic multiplicity of the eigenvalue $k_0$ of $A$. 
\end{corollary}

The above results reveal how a perturbation on a holomorphic Fredholm operator-valued function changes their eigenvalue distribution locally. For example, suppose that a holomorphic Fredholm operator-valued function $A:\Lambda\to\mathcal L(X,Y)$ has a simple eigenvalue at $k_0$, i.e., both the geometric and algebraic multiplicities of $k_0$ are one. Then for every $B\in\mathcal L(X,Y)$, an eigenvalue of the perturbed function $A+\varepsilon B$ in a neighborhood of $k_0$ can be identified as a zero of $k\mapsto k - k_0 + c_0(\varepsilon)$, where $c_0$ is a holomorphic function with $c_0(0)=0$, for sufficiently small $\varepsilon\in\mathbb C$. This implies that the eigenvalue moves continuously in the neighborhood in response to the small change of $\varepsilon$.

Let us consider a defective eigenvalue $k_0$ of $A:\Lambda\to\mathcal L(X,Y)$ with geometric multiplicity of one and algebraic multiplicity of two. Then for sufficiently small $\varepsilon$, the perturbed function $A+\varepsilon B$ has eigenvalues at the zeros of $k\mapsto (k-k_0)^2 + c_1(\varepsilon)(k-k_0) + c_0(\varepsilon)$, where $c_0$ and $c_1$ are holomorphic functions with $c_0(0)=c_1(0)=0$. The zeros are explicitly written as
\begin{align*}
    k^{(1)}(\varepsilon) &:= k_0 + \frac{-c_1(\varepsilon) + \sqrt{(c_1(\varepsilon))^2 - 4c_0(\varepsilon) }}{2},
\\
    k^{(2)}(\varepsilon) &:= k_0 + \frac{-c_1(\varepsilon) - \sqrt{(c_1(\varepsilon))^2 - 4c_0(\varepsilon) }}{2}.
\end{align*}
This means that the defective eigenvalue $k_0$ of $A$ splits into two eigenvalues $k^{(1)}(\varepsilon)$ and $k^{(2)}(\varepsilon)$ with $k^{(1)}(\varepsilon)\neq k^{(2)}(\varepsilon)$ for all $\varepsilon\neq 0$ in the neighborhood, provided that $c_1^2 - 4c_0$ is not identically zero. This behavior is schematically illustrated in \cref{fig:traj}. In other words, arbitrary small perturbation $\varepsilon$ changes the number of distinct eigenvalues in the neighborhood. Such parameter value $\varepsilon=0$ is called an \textit{exceptional point} of the family $\{A+\varepsilon B\}_{\varepsilon\in V}$. 

The Taylor expansion $(c_1(\varepsilon))^2 - 4c_0(\varepsilon) = -4c_0^\prime(0)\varepsilon + O(\varepsilon^2)$ at $\varepsilon=0$ shows that there exists a constant $\alpha\in \mathbb C$ such that
\begin{align}
    k^{(1)}(\varepsilon) = k_0 + \alpha \varepsilon^{1/2} + O(\varepsilon), \quad k^{(2)}(\varepsilon) = k_0 - \alpha \varepsilon^{1/2} + O(\varepsilon). \label{eq:sqrt}
\end{align}
The constant $\alpha$ vanishes if and only if $c_0^\prime(0) = 0$. 

In general, the zeros $k^{(1)},\ldots,k^{(p)}$ of the polynomial $z\mapsto W(\varepsilon,z)$ can be written in the following Puiseux series \cite{nennig2020high}:
\begin{align*}
    k^{(j)}(\varepsilon) = k_0 + \alpha^{(j)}_1\varepsilon^{1/p} + \alpha^{(j)}_2\varepsilon^{2/p} + \ldots + \alpha^{(j)}_{p-1}\varepsilon^{(p-1)/p} + O(\varepsilon) \quad\text{for all }j=1,\ldots,p
\end{align*}
with some constants $\alpha^{(j)}_1,\ldots,\alpha^{(j)}_{p-1}\in\mathbb C$.

\begin{figure}
    \centering
    \includegraphics[scale=0.63]{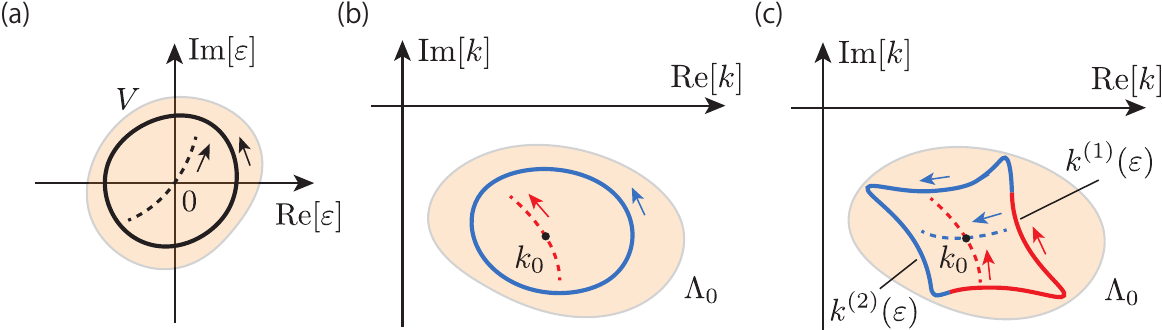}
    \caption{Schematic illustration of eigenvalues of $A+\varepsilon B$ as a function of $\varepsilon$, where $A$ has an eigenvalue $k_0$ with geometric multiplicity $m$ and algebraic multiplicity $p$. (a) Two paths (solid and dashed lines) in the neighborhood $V$. (b) and (c) Trajectory of eigenvalues in $\Lambda_0$ when $\varepsilon$ continuously changes along the paths. (b) $m=p=1$. (c) $m=1$ and $p=2$.}
    \label{fig:traj}
\end{figure}

\section{Numerical method}\label{s:numerical}
We wish to compute approximation to resonances for a given geometry $\Omega$ numerically. To this end, the boundary integral equation $(I - D_k)\varphi = 0$ is discretized into a finite-dimensional linear system $A_N(k) \Psi_N = 0$, where $A_N$ is a square matrix, such that $\Psi_N$ approximates $\varphi\in C(\partial\Omega)$ in some sense. We adopt a numerical scheme based on the Nystr\"om discretization, which is recently proposed for Dirichlet and Robin (Neumann) problems with simply-connected obstacles \cite{ma2023computation,liu2025accurate}.

\subsection{Nystr\"om method}
We use the Nystr\"om method proposed by \cite{kress1991boundary} as summarized below.
Let $\Omega^{(1)},\ldots,\Omega^{(M)}$ be subdomains of $\Omega$ such that $\Omega^{(m)}$ is simply connected for all $m=1,\ldots,m$, $\Omega^{(i)}$ and $\Omega^{(j)}$ are disjoint for all $i\neq j$, and $\Omega = \cup_{m=1}^M \Omega^{(m)}$. Then for every $m=1,\ldots,M$ there exists an embedding $\gamma^{(m)}:\mathbb R/(2\pi\mathbb Z)\to\mathbb R^2$ of class $C^2$ such that $\partial\Omega^{(m)} = \{ (\gamma^{(m)}_1(t),\gamma^{(m)}_2(t)) : t\in[0,2\pi] \}$ and its outward unit normal $\nu$ on $\partial\Omega^{(m)}$ is given by $\nu^{(m)}(t) := (\gamma^{(m)\prime}_2(t),-\gamma^{(m)\prime}_1(t))/|\gamma^{(m)\prime}(t)|$ at $\gamma^{(m)}(t)$ for each $t\in[0,2\pi]$. This parameterization allows us to rewrite the boundary integral equation $(I - D_k)\varphi = 0$ as: find $(\psi^{(1)},\ldots,\psi^{(M)})\in (C(\mathbb R/(2\pi\mathbb Z)))^M$ such that
\begin{align*}
    \psi^{(m)}(t) - \sum_{l=1}^M \int_0^{2\pi} K_{ml}(t,\tau;k) \psi^{(l)}(\tau)\mathrm d\tau = 0 \quad\text{for all }t\in[0,2\pi] \text{ and }m=1,\ldots,M,
\end{align*}
where the kernel $K_{ml}$ is given by
\begin{align*}
    K_{ml}(t,\tau;k) = \frac{\mathrm i k}{2} H^{(1)}_1(k|\gamma^{(m)}(t)-\gamma^{(l)}(\tau)|) \frac{\nu^{(l)}(\tau)\cdot(\gamma^{(m)}(t)-\gamma^{(l)}(\tau))}{|\gamma^{(m)}(t)-\gamma^{(l)}(\tau)|} |\gamma^{(l)\prime}(\tau)| 
\end{align*}
for all $t,\tau\in[0,2\pi]$ and $m,l=1,\ldots,M$. For numerical quadratures introduced later, the equation is further rewritten as
\begin{align}
\label{eq:ie}
    \psi^{(m)}(t) - &\left(\sum_{l\neq m} \int_0^{2\pi} K_{ml}(t,\tau;k) \psi^{(l)}(\tau)\mathrm d\tau + \int_0^{2\pi} \hat K_{m}(t,\tau;k) \psi^{(m)}(\tau)\mathrm d\tau \right)
    \\
    &-\int_0^{2\pi} \tilde K_{m}(t,\tau;k) \log\left(4\sin^2\frac{t-\tau}{2}\right) \psi^{(m)}(\tau)\mathrm d\tau  = 0  \notag
\end{align}
for all $t\in[0,2\pi]$ and $m=1,\ldots,M$, where the kernels $\tilde K_m$ and $\hat K_m$ are defined by
\begin{align*}
    \tilde K_m(t,\tau;k) &= \frac{\mathrm i k}{2} J_1(k|\gamma^{(m)}(t)-\gamma^{(m)}(\tau)|) \frac{\nu^{(m)}(\tau)\cdot(\gamma^{(m)}(t)-\gamma^{(m)}(\tau))}{|\gamma^{(m)}(t)-\gamma^{(m)}(\tau)|} |\gamma^{(m)\prime}(\tau)|,
\\
    \hat K_m(t,\tau;k) &= K_m(t,\tau;k) - \tilde K_m(t,\tau;k) \log\left( 4\sin^2 \frac{t-\tau}{2} \right)
\end{align*}
for all $t,\tau\in[0,2\pi]$ and $m=1,\ldots,M$. The original equation $(I-D_k)\varphi = 0$ on $C(\partial\Omega)$ and integral equation \cref{eq:ie} on $(C(\mathbb R/(2\pi\mathbb Z)))^M$ are equivalent in the sense that their unique solvability is the same and their solutions can be converted by the formula $\varphi\circ \gamma^{(m)} = \psi^{(m)}$ for all $m=1,\ldots,M$.

To construct a numerical solution of \cref{eq:ie}, we first approximate the integrals in the second term on the left-hand side of \cref{eq:ie} by the $2N$-points trapezoidal rule, i.e., 
\begin{align*}
    &\sum_{l\neq m} \int_0^{2\pi} K_{ml}(t,\tau;k) \psi^{(l)}(\tau)\mathrm d\tau + \int_0^{2\pi} \hat K_{m}(t,\tau;k) \psi^{(m)}(\tau)\mathrm d\tau 
\\
    &\simeq \sum_{l\neq m} \frac{2\pi}{N} \sum_{j=0}^{2N-1} K_{ml}(t,t_j;k) \psi^{(l)}(t_j) + \frac{2\pi}{N} \sum_{j=0}^{2N-1} \hat K_{m}(t,t_j;k) \psi^{(m)}(t_j)
\end{align*}
for each $t\in[0,2\pi]$ and $m=1,\ldots,M$, where the quadrature points are given by $t_j=j\pi/N$ ($j=0,\ldots,2N-1)$. The integral in the third term on the left-hand side of \cref{eq:ie} is approximated using the Kussmaul--Martensen rule \cite{kussmaul1969ein,martensen1963uber}, given by
\begin{align*}
    \int_0^{2\pi} \tilde K_{m}(t,\tau;k) \log\left(4\sin^2\frac{t-\tau}{2}\right) \psi^{(m)}(\tau)\mathrm d\tau \simeq \frac{2\pi}{N} \sum_{j=0}^{2N-1} R_j(t) \tilde K_m(t,t_j;k) \psi^{(m)}(t_j)
\end{align*}
for each $t\in[0,2\pi]$ and $m=1,\ldots,M$, where $R_j(t)$ is defined by 
\begin{align*}
    R_j(t) =& - \frac{\pi}{N^2} \cos N(t-t_j) -\frac{2\pi}{N} \sum_{m=1}^{N-1}\frac{1}{m}\cos m(t-t_j).
\end{align*}
We apply the Nystr\"om method based on these quadratures and solve the linear system: find $\Psi^{(1)},\ldots,\Psi^{(M)}\in \mathbb C^{2M}$ such that
\begin{align}
    \begin{bmatrix}
        A^{(1,1)}(k) & \cdots & A^{(M,1)}(k)
        \\
        \vdots & \ddots & \vdots
        \\
        A^{(1,M)}(k) & \cdots & A^{(M,M)}(k)
    \end{bmatrix}
    \begin{pmatrix}
        \Psi^{(1)} \\ \vdots \\ \Psi^{(M)}
    \end{pmatrix}
    =
    \begin{pmatrix}
        0 \\ \vdots \\ 0
    \end{pmatrix}. \label{eq:discrete}
\end{align}
For each $m,l=1,\ldots,M$, the $2N\times 2N$ matrix $A^{(m,l)}$ is defined by
\begin{align*}
    A^{(m,m)}_{ij}(k) &= \delta_{ij} - \frac{2\pi}{N} \hat K_{m}(t_i,t_j;k) - \frac{2\pi}{N}  R_j(t_i) \tilde K_m(t_i,t_j;k) \quad \text{for all }m=1,\ldots,M,
\\
    A^{(m,l)}_{ij}(k) &= -\frac{2\pi}{N} K_{ml}(t_i,t_j;k) \quad \text{for all }m,l=1,\ldots,M \text{ with }m\neq l.
\end{align*}
A solution of the linear system \cref{eq:discrete} gives approximation to a solution of the integral equation \cref{eq:ie} as $\Psi^{(m)}_i\simeq \psi^{(m)}(t_i)$ for each $m=1,\ldots,M$ and $i=0,\ldots,2N-1$. If the boundary $\partial\Omega$ is analytic, i.e., $\gamma^{(1)},\ldots,\gamma^{(M)}$ are analytic functions, then we expect that this numerical scheme is spectrally accurate \cite{kress2014linear}.

\subsection{Sakurai--Sugiura method}
Our aim is to numerically calculate eigenvalues of $k\mapsto I-D_k$, i.e., values of $k$ for which the integral equation \cref{eq:ie} admits a nontrivial solution in $(C(\mathbb R/(2\pi\mathbb Z)))^M$. Based on the Nystr\"om method, we expect that these eigenvalues are approximately given by the spectrum of $k\mapsto A_N(k)$, where $A_N(k)$ is the matrix on the right-hand side of \cref{eq:discrete}, for sufficiently large $N$.

As $A_N(k)$ is a square matrix of finite size, an eigenvalue of $k\mapsto A_N(k)$ is the value of $k$ for which the determinant of $A_N(k)$ vanishes. Such values can be numerically sought using, for example, the Newton method. In this study, however, we employ a more advanced numerical method, known as Sakurai--Sugiura method \cite{asakura2009numerical}, as it can find multiple eigenvalues in a given domain of $\mathbb C$. For the reader's convenience, we briefly summarize the algorithm of the Sakurai--Sugiura method below.

We want to find eigenvalues of the holomorphic matrix-valued function $k\mapsto A_N(k)$ from an open and connected subset $\Lambda$ of $\mathbb C$ onto $\mathcal L(\mathbb C^{2NM},\mathbb C^{2NM})$. Let $\Gamma$ be a positively oriented closed Jordan curve in $\Lambda$ and suppose that the matrix pencil $k\mapsto H^<-k H$ has $N_\Gamma$ eigenvalues (counting multiplicity) $k_1,\ldots,k_{N_\Gamma}$ inside $\Gamma$ with algebraic multiplicity less than or equal to $L_\Gamma$, where $H$ and $H^<$ are the block Hankel matrices defined by
\begin{align*}
    H = 
    \begin{bmatrix}
        \mu_0 & \mu_1 & \cdots & \mu_{N_\Gamma-1}
        \\
        \mu_1 & \mu_2 & \cdots & \mu_{N_\Gamma}
        \\
        \vdots & \vdots & \ddots & \vdots
        \\
        \mu_{N_\Gamma-1} & \mu_{N_\Gamma} & \cdots & \mu_{2N_\Gamma-2}
    \end{bmatrix}
    ,\quad
    H^< = 
    \begin{bmatrix}
        \mu_1 & \mu_2 & \cdots & \mu_{N_H}
        \\
        \mu_2 & \mu_3 & \cdots & \mu_{N_H+1}
        \\
        \vdots & \vdots & \ddots & \vdots
        \\
        \mu_{N_H} & \mu_{N_H+1} & \cdots & \mu_{2N_H-1}
    \end{bmatrix}
    ,
\end{align*}
and $\mu_j$ are defined by
\begin{align*}
    \mu_j = \int_\Gamma k^j \mathcal U^H A_N^{-1}(k) \mathcal V \mathrm dk
\end{align*}
with $2NM$-by-$L_\Gamma$ full column-rank matrices $\mathcal U$ and $\mathcal V$. Then $A_N$ has eigenvalues at $k_1,\ldots,k_{N_\Gamma}$. Under some assumptions on $\mathcal U$ and $\mathcal V$, the converse also holds, i.e., $k_1,\ldots,k_{N_\Gamma}$ are eigenvalues of $k\mapsto H^<-k H$ if they are eigenvalues of $A_N$. This equivalence allows us to compute eigenvalues of $A_N$ by solving the generalized eigenvalue problem for $(H^<,H)$, which can easily be solved by standard methods in numerical linear algebra.

\section{Numerical results}\label{s:result}
In this section, we present some numerical results suggesting the existence of a defective resonance in the exterior Neumann problem. Throughout this section, we set the scatterer $\Omega$ as the union of a finite number of open disks given by
\begin{align*}
    \Omega = \left(\bigcup_{j=1}^{20} \mathcal B(c_{1j};R_1) \right) \cup \left(\bigcup_{j=1}^{20} \mathcal B(c_{2j};R_2) \right),
\end{align*}
where $R_1$ and $R_2\in(0,1/2)$ are constants with $c_{1j} := (j-1,1/2)\in\mathbb R^2$ and $c_{2j} := (j-1,-1/2)\in\mathbb R^2$ for all $j=1,\ldots,20$. This geometry is schematically shown in \cref{fig:geometry}. 
\begin{figure}
    \centering
    \includegraphics[scale=0.63]{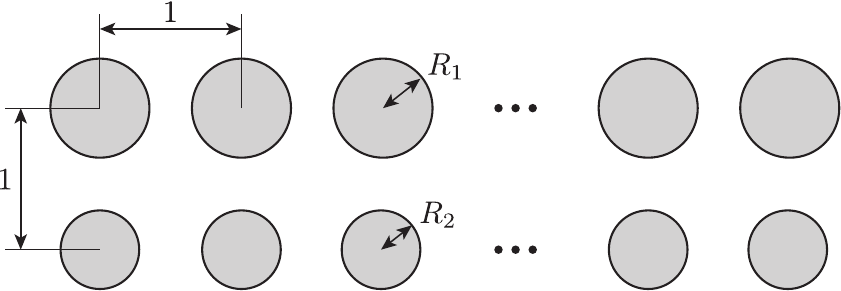}
    \caption{$20\times 2$ disks aligned on the two-dimensional Cartesian grid. The top and bottom disks have the radii $R_1$ and $R_2$, respectively.}
    \label{fig:geometry}
\end{figure}

Each line array of the sound-hard scatterers supports guided or leaky waves, called the Rayleigh--Bloch waves \cite{matsushima2024tracking}. We expect that two Rayleigh--Bloch waves propagating along top and bottom arrays interact with each other and yield a defective resonance. This mechanism is discussed in Appendix A using a simplified model. While the Rayleigh--Bloch wave theory is valid in purely periodic systems with infinite length, it provides some analogy between infinite and finite arrays \cite{thompson2008new,chaplain2025acoustic}.

\subsection{Finding a degenerate resonance via optimization}
To seek a defective resonance, we find $(R_1,R_2)$ such that two simple resonances $k_1$ and $k_2$ coalesce at a single point in the complex $k$ plane. This is formulated as the minimization of the modulus $|k_1(R_1,R_2)-k_2(R_1,R_2)|$ in terms of the two real variables $R_1$ and $R_2$ \cite{matsushima2023exceptional}. We numerically seek a minimizer using the Nelder--Mead method, a well-known gradient-free optimization algorithm.

%
%
The algorithm found two resonances
\begin{align*}
    k^{(1)} &:= 4.66118535 - 0.38585151 \mathrm i
\\
    k^{(2)} &:= 4.66118551 - 0.38585143 \mathrm i
\end{align*}
at $(R_1,R_2) = (0.3500359767278,0.2777547442949)$.
The number of quadrature points is set to $N=50$. Each resonance has a single unique eigenvector. The two resonances are almost degenerate around the point $k_0 := 4.661185 - 0.385851 \mathrm i$. We assume that these resonances are in fact degenerate and estimate the geometric multiplicity numerically. This can be done by calculating the column rank of the matrix $[\Psi^{(1)},\Psi^{(2)}]$, where $\Psi^{(1)}$ and $\Psi^{(2)}$ respectively denote eigenvectors of the discretized system corresponding to $k^{(1)}$ and $k^{(2)}$. The rank is equal to the number of nonzero singular values of the matrix. 
We calculated two singular values $\sigma_1\geq \sigma_2$ and obtained the value $\sigma_2/\sigma_1 \simeq 6.43\times 10^{-7}$. 
This suggests that the two eigenvectors are almost linearly dependent, i.e., the geometric multiplicity of the degenerate resonance $k_0$ is one.

\begin{figure}
    \centering
    \includegraphics[width=1\linewidth]{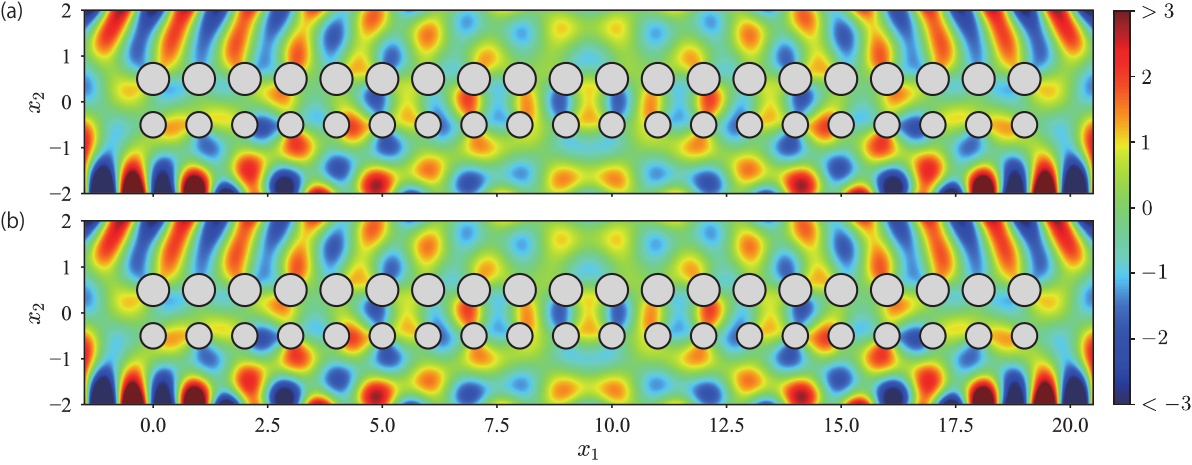}
    \caption{(a) and (b) The real part of resonant states $u^{(1)}$ and $u^{(2)}$ corresponding to the two resonances $k^{(1)}$ and $k^{(2)}$, respectively. The values are scaled such that $u^{(1)}(9.5,0)=u^{(2)}(9.5,0)=1$.}
    \label{fig:result-naiten}
\end{figure}
In addition, we calculated their corresponding resonant states and plotted their distribution in \cref{fig:result-naiten}. This result clearly indicates that the two resonant states are identical up to constant factors.

\begin{table}[]
\centering
\caption{Values of two resonances $k^{(1)}$ and $k^{(2)}$ for various $N$.}
\label{tab:result-change-N}
\begin{tabular}{lllll}
$N$  & $\mathrm{Re}[k^{(1)}]$ & $\mathrm{Im}[k^{(1)}]$ & $\mathrm{Re}[k^{(2)}]$ & $\mathrm{Im}[k^{(2)}]$ \\
\hline
$50$ & $4.661185359$ & $-0.385851510$ & $4.661185516$ & $-0.385851429$
\\
$70$ & $4.661185369$ & $-0.385851498$ & $4.661185506$ & $-0.385851441$
\\
$90$ & $4.661185358$ & $-0.385851508$ & $4.661185516$ & $-0.385851432$
\end{tabular}
\end{table}
To confirm that the degeneracy is not a numerical artifact originating from the discretization, we changed the discretization parameter $N$ and calculated resonances around the point $k_0$. The results are shown in \cref{tab:result-change-N}. We find that two resonances exist and they are almost degenerate around $k_0$ with modulus $|k^{(1)}-k^{(2)}|$ less than $2\times 10^{-7}$. This suggests that the degeneracy is not a spurious one due to the discretization. In the following computation, the parameter is fixed at $N=50$.

%
%
\subsection{Length of Jordan chains}
To confirm that the resonance $k_0$ is degenerate with the algebraic multiplicity greater than one, we numerically estimate the maximum length of Jordan chains at $k_0$.

Let $\varphi_0\in C(\partial\Omega)$ be the unique eigenvector associated with the degenerate resonance $k_0$. By definition, the pair $(\varphi_0,\varphi_1)$ is a Jordan chain if and only if $\varphi_1\in C(\partial\Omega)$ solves
\begin{align}
    (I - D_{k_0})\varphi_1 = D_{k_0}^\prime \varphi_0, \label{eq:Jordan}
\end{align}
where $D_k^\prime\in\mathcal L(C(\partial\Omega),C(\partial\Omega))$ is the derivative of $k\mapsto D_k$. Let $\langle \cdot,\cdot\rangle:C(\partial\Omega)\times C(\partial\Omega)\to\mathbb C$ be the nondegenerate bilinear form defined by
\begin{align*}
    \langle \varphi,\psi\rangle = \int_{\partial\Omega} \varphi(x)\psi(x)\mathrm ds \quad\text{for all }\varphi,\psi\in C(\partial\Omega).
\end{align*}
Then the boundary integral operator $D_{k_0}$ has a unique adjoint operator $D_{k_0}^*$ with respect to $\langle \cdot,\cdot\rangle$, given by
\begin{align*}
    (D^*_{k_0} \varphi)(x) = \int_{\partial\Omega} \frac{\partial \Phi}{\partial \nu(x)}(x,y;k_0)\varphi(y) \mathrm ds(y) \quad\text{for all }\varphi\in C(\partial\Omega), x\in\partial\Omega.
\end{align*}
By the Fredholm alternative, the equation \cref{eq:Jordan} is solvable if and only if 
\begin{align*}
    \langle D_{k_0}^\prime \varphi_0, \psi \rangle = 0 \quad\text{for all }\psi\in N(I-D_{k_0}^*).
\end{align*}
This condition can be tested numerically by computing an eigenvector of $k\mapsto I-D_k^*$ and $D^\prime_{k_0}\varphi_0$. We apply the same Nystr\"om method and compute an approximation of $\langle D_{k_0}^\prime \varphi_0, \psi \rangle$.

We first checked the norm of $D_{k_0}^\prime \varphi_0$ and obtained $\|D_{k_0}^\prime \varphi_0\|/\| \varphi_0\| \simeq 2.34\times 10^2$, implying $D_{k_0}^\prime \varphi_0\neq 0$. In addition, we calculated the value of $|\langle D_{k_0}^\prime \varphi_0/\| \varphi_0\|, \psi/\|\psi\| \rangle|$ for the degenerate resonance at $k_0$ and another simple resonance at $k_1:=4.697549 -0.391284\mathrm i$ and obtained
\begin{align*}
    |\langle D_{k_0}^\prime \varphi_0/\| \varphi_0\|, \psi/\|\psi\| \rangle| \simeq 8.57 \times 10^{-6} \quad \text{and}\quad |\langle D_{k_1}^\prime \varphi_0/\| \varphi_0\|, \psi/\|\psi\| \rangle| \simeq 3.61.
\end{align*}
The degenerate resonance $k_0$ has a much smaller value of $\langle D_{k_0}^\prime \varphi_0, \psi \rangle$ in comparison with the other simple resonance. This suggests that the equation \cref{eq:Jordan} is solvable, i.e. the algebraic multiplicity of $k_0$ is greater than one.

%
%
\subsection{Exceptional point}
As we discussed in Section 3(b), the existence of a defective resonance induces an exceptional point of a family of operator-valued functions. To see this numerically, let us consider a compact operator $B:C(\partial\Omega)\to C(\partial\Omega)$ given by
\begin{align*}
    (B\varphi)(x) = \int_{\partial\Omega} \sin|x-y| \varphi(y) \mathrm ds(y) \quad\text{for all }\varphi\in C(\partial\Omega), x\in\partial\Omega.
\end{align*}

\begin{figure}
    \centering
    \includegraphics[width=0.95\linewidth]{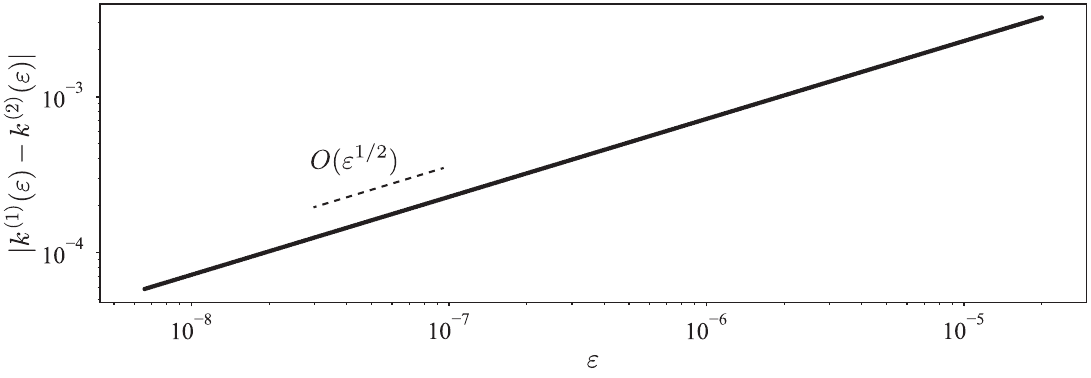}
    \caption{Distance between the two resonances $k^{(1)}(\varepsilon)$ and $k^{(2)}(\varepsilon)$ as a function of $\varepsilon$.}
    \label{fig:camphor-sqrt_result}
\end{figure}
We calculated two eigenvalues $k^{(1)}(\varepsilon)$ and $k^{(2)}(\varepsilon)$ of $A+\varepsilon B$ around $k_0$ for varying $\varepsilon$ on the real axis and plotted their distance in \cref{fig:camphor-sqrt_result}. This result implies that the distance $|k^{(1)}(\varepsilon) - k^{(2)}(\varepsilon)|$ changes at the rate of $\varepsilon^{1/2}$, which is consistent with the formula \cref{eq:sqrt}.

This square root characteristics becomes more evident when the parameter $\varepsilon$ varies along a closed path encircling the point $\varepsilon=0$. As shown in \cref{fig:camphor-encircle_result} (a), we consider the path $\varepsilon=2\times 10^{-5}\mathrm e^{\mathrm i\theta}$ with $\theta$ changing from $0$ to $2\pi$. The corresponding eigenvalues of $A+\varepsilon B$ are calculated and plotted in \cref{fig:camphor-encircle_result} (b). The two eigenvalues continuously move on the complex $k$-plane counterclockwise and switch their positions once a single encircling (from $\theta=0$ to $2\pi$ in the figure) is completed. This unique behavior, known as \textit{exceptional point encircling} \cite{dembowski2004encircling}, suggests that the perturbed eigenvalues obey the square-root asymptotics \cref{eq:sqrt} and also $\varepsilon=0$ is an exceptional point of the family $\{ A+\varepsilon B\}_\varepsilon$.

\begin{figure}
    \centering
    \includegraphics[width=0.95\linewidth]{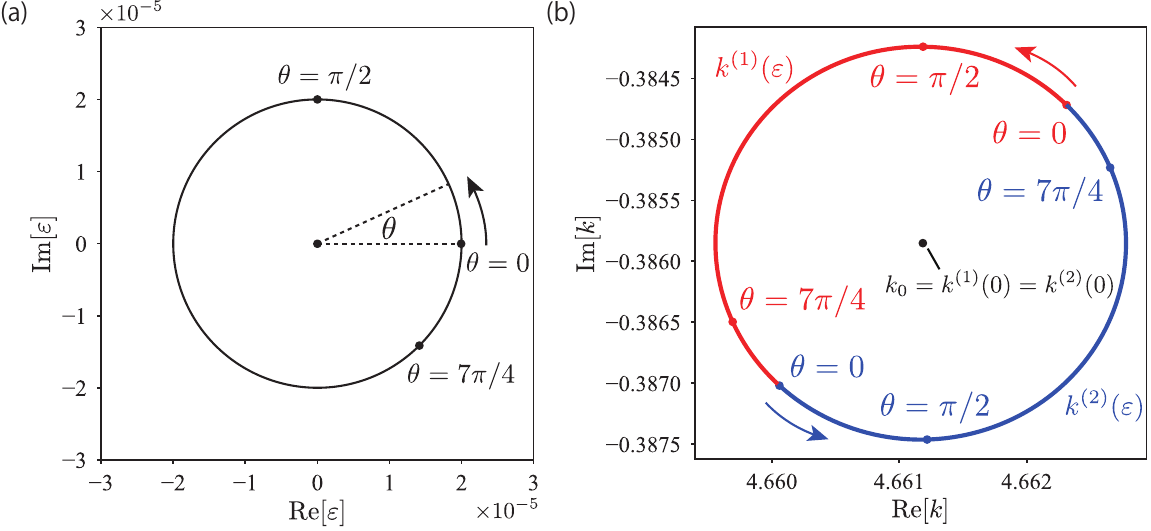}
    \caption{(a) Positively oriented circle of radius $2\times 10^{-5}$ centered at the origin. (b) Trajectories of the two resonances $k^{(1)}(\varepsilon)$ and $k^{(2)}(\varepsilon)$ when $\varepsilon$ moves along the circle.}
    \label{fig:camphor-encircle_result}
\end{figure}

\section{Concluding remarks}\label{s:conclusion}
In this paper, we studied resonances in an acoustic scattering problem with the sound-hard boundary condition. A boundary integral formulation allowed us to employ perturbation theory for holomorphic Fredholm operator-valued functions, which reveals some unique characteristics of defective resonances. Some numerical results based on the Nystr\"om method showed strong evidence of the existence of a defective resonance in the sound-hard system. This suggests that conventional high-contrast material properties are not essential for achieving an exceptional point. 

Although our numerical scheme is valid only for the two-dimensional setting, the theoretical analysis on operator-valued functions can be applied to higher-dimensional cases. One future direction of this work is to check whether defective resonances still exist in three-dimensional scattering systems. In addition, we expect that other numerical techniques, including finite element methods with the Dirichlet-to-Neumann map \cite{xi2024analysis}, are useful for the analysis of defective resonances. We also mention that the approach based on the Gohberg--Sigal theory \cite{ammari2021high} may be applicable to our problem.

Note that the physical origin of the obtained degeneracy and exceptional point is unclear. The analysis on a simplified model in Appendix A suggests the connection between the existence of exceptional points and the Bloch band structure in a periodic system. More detailed analysis and discussion are expected in future work. If this is the case, the mechanism of the defective degeneracy can be fundamentally different from that in parity--time-symmetric scattering systems \cite{ammari2022exceptional}. The code used to conduct this study is available at \cite{zenodo}.

\section*{Acknowledgments}
K.M. is supported by JSPS KAKENHI (JP24K17191, JP23H03798, JP23H03413) and Mizuho Foundation for the Promotion of Sciences.
This work is supported by ``Joint Usage/Research Center for Interdisciplinary Large-scale Information Infrastructures (JHPCN)” in Japan (Project IDs: jh240031 and jh250045). In this research work we used the supercomputer of ACCMS, Kyoto University. The authors acknowledge S. Kashimura for his contribution to the numerical results presented in this paper and K. Morishita for helpful advice.

\appendix

\section{Exceptional point in a mechanical system}\label{appendix}
\begin{figure}
    \centering
    \includegraphics[width=1\linewidth]{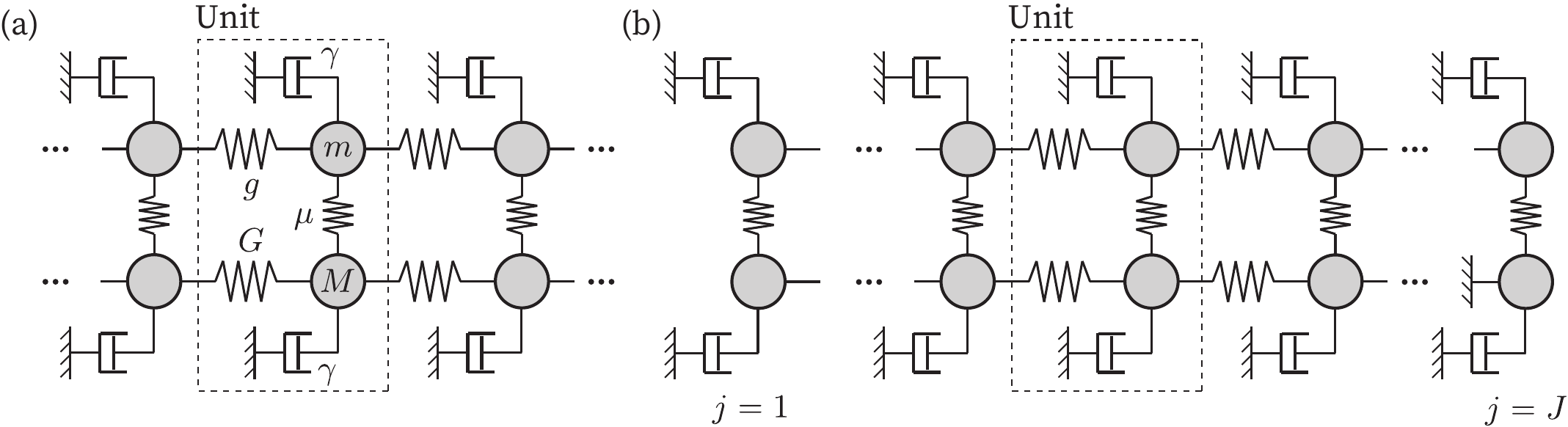}
    \caption{The infinite-length chain (a) has a unit component with two point masses, three springs, and two dampers. The finite-length chain (b) has the same unit component.}
    \label{fig:mechanical-chain}
\end{figure}
As shown in \cref{fig:mechanical-chain} (a), let us consider an infinite-length chain of point masses $m$ and $M$ connected by springs with positive constants $g$, $G$ and $\mu$. In addition, each point mass is attached to a damper with constant $\gamma>0$. 
We consider the horizontal displacement of each mass in the chain under the time-harmonic assumption with angular frequency $\omega$. Let $u$ and $U$ be the displacement of the point mass $m$ and $M$ in a unit, respectively. They satisfy the following linear system:
\begin{align*}
    \begin{bmatrix}
        2g(1-\cos\beta)+\mu-\mathrm i\omega\gamma-m\omega^2 & -\mu
        \\
        -\mu & 2G(1-\cos\beta)+\mu-\mathrm i\omega\gamma-M\omega^2
    \end{bmatrix}
    \begin{pmatrix}
        u \\ U
    \end{pmatrix}
    =
    \begin{pmatrix}
        0 \\ 0
    \end{pmatrix}
    ,
\end{align*}
where $\beta\in[-\pi,\pi]$ is the wavenumber in the horizontal direction. This quadratic eigenvalue problem is equivalent to the following linear eigenvalue problem:
\begin{align*}
    \mathcal T(\beta)
    \begin{pmatrix}
        u \\ U \\ \omega u \\ \omega U
    \end{pmatrix}
    =
    \omega
    \begin{pmatrix}
        u \\ U \\ \omega u \\ \omega U
    \end{pmatrix}
    ,
\end{align*}
where $\mathcal T(\beta)$ is the four-by-four matrix given by
\begin{align*}
    \mathcal T(\beta) =
    \begin{bmatrix}
        1 & & & 
    \\
          & 1 & &
    \\
        & & 1/m &
    \\
        & & & 1/M
    \end{bmatrix}
    \begin{bmatrix}
         &  & 1 & 
    \\
         &  &  & 1
    \\
        2g(1-\cos\beta)+\mu & -\mu & -\mathrm i\gamma & 
    \\
        -\mu & 2G(1-\cos\beta)+\mu &  & -\mathrm i\gamma
    \end{bmatrix}
    .
\end{align*}
Therefore, the Bloch band of the periodic chain is given by
\begin{align*}
    \bigcup_{\beta\in[-\pi,\pi]} \sigma(\mathcal T(\beta)).
\end{align*}
An important observation is that the square matrix $\mathcal T(\beta)$ may possess defective eigenvalues for certain parameters. For example, let us consider the following set of parameters: $m=g=1$, $M=2$, $\gamma=8/3$, $\mu=4\sqrt{5}/9$, and $G=2(81+\sqrt{5})/149$. In the case of $\beta=\cos^{-1} \frac{2\sqrt 5-4}{9}$, the matrix $\mathcal T(\beta)$ has two defective eigenvalues at $\pm 1-\mathrm i$ with the Jordan decomposition
\begin{align*}
    \mathcal T(\beta) = 
    \mathcal S
\begin{bmatrix}
-1 - \mathrm i & 1 & 0 & 0\\
0 & -1 - \mathrm i & 0 & 0\\
0 & 0 & 1 - \mathrm i & 1\\
0 & 0 & 0 & 1 - \mathrm i
\end{bmatrix}
\mathcal S^{-1},\quad
\mathcal S := 
\begin{bmatrix}
\displaystyle\frac{1+2i}{\sqrt{5}} & \displaystyle-\frac{9-7i}{2\sqrt{5}} & \displaystyle-\frac{1-2i}{\sqrt{5}} & \displaystyle-\frac{9+7i}{2\sqrt{5}}\\[4pt]
\displaystyle-\frac{1-i}{2}       & \displaystyle\frac{i}{2}            & \displaystyle\frac{1+i}{2}         & \displaystyle-\frac{i}{2}\\[4pt]
\displaystyle\frac{1-3i}{\sqrt{5}} & \displaystyle\frac{9+3i}{\sqrt{5}} & \displaystyle\frac{1+3i}{\sqrt{5}} & \displaystyle-\frac{9-3i}{\sqrt{5}}\\[4pt]
1 & 0 & 1 & 0
\end{bmatrix}.
\end{align*}
The Bloch band is plotted in \cref{fig:result-analogy_result} (a). This suggests that the defective degeneracy of the eigenvalues $\omega=\pm 1-\mathrm i$ is associated with the crossing of two continuous dispersion curves in the three-dimensional $(\beta,\omega)$ space.

\begin{figure}
    \centering
    \includegraphics[width=1\linewidth]{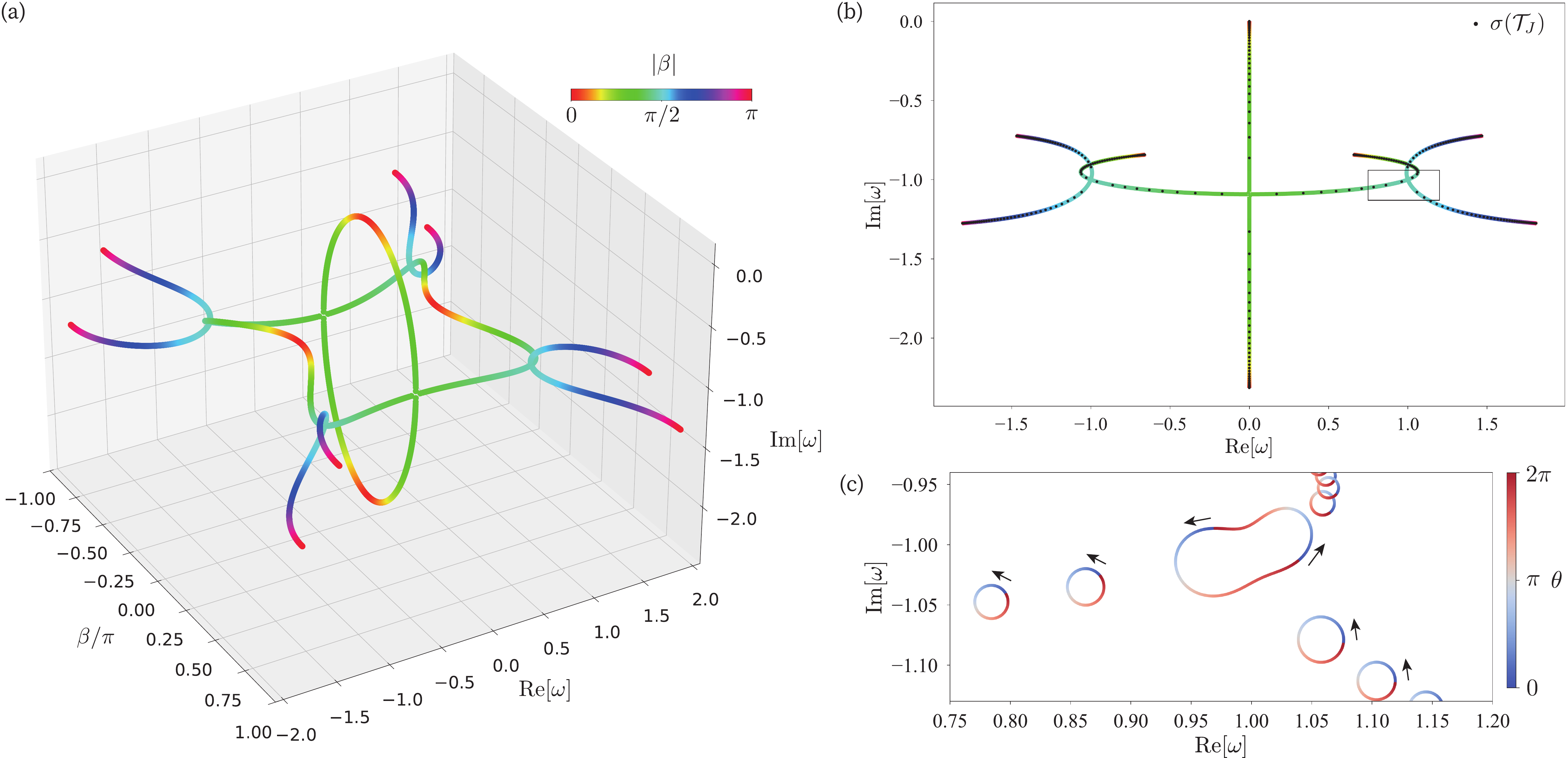}
    \caption{Spectra of the matrices $\mathcal T(\beta)$, $\mathcal T_J$, and $\mathcal T_J^\varepsilon(\theta)$. The parameters are set to $m=g=1$, $M=2$, $\gamma=8/3$, $\mu=4\sqrt{5}/9$, $G=2(81+\sqrt{5})/149$, and $J=80$. (a) Spectrum $\sigma(\mathcal T(\beta))\subset \mathbb C$ for varying $\beta\in[-\pi,\pi]$. (b) Comparison of $\cup_{\beta\in[-\pi,\pi]}\,\sigma(\mathcal T(\beta))$ (continuous line) and $\sigma(\mathcal T_J)$ (dots). (c) Trajectory of some eigenvalues of $\mathcal T^\varepsilon_{J}(\theta)$ for $\varepsilon=8\times 10^{-3}$ and $\theta\in[0,2\pi]$. The arrows indicate the direction of increasing $\theta$.}
    \label{fig:result-analogy_result}
\end{figure}
We show that the defective degeneracy of eigenvalues of $\mathcal T(\beta)$ is associated with the existence of exceptional points in finite-length systems with open boundary conditions. For instance, let us consider the finite-length chain shown in \cref{fig:mechanical-chain} (b). The finite-length system is characterized by eigenvalues of the $4J$-by-$4J$ matrix
\begin{align*}
    \mathcal T_J =
    \begin{bmatrix}
        I & O
        \\
        O & \mathcal M_J
    \end{bmatrix}^{-1}
    \begin{bmatrix}
        O & I 
        \\
        \mathcal K_J & -\mathrm i\gamma I
    \end{bmatrix},
\end{align*}
where $J$ is the length of the chain and the $2J$-by-$2J$ blocks $\mathcal M_J$ and $\mathcal K_J$ are defined by
\begin{align*}
    \mathcal M_J &= \mathrm{diag}(m,m,\ldots,m,M,M,\ldots,M)
    ,
    \\
    \mathcal K_J &=
    \left[
    \begin{array}{@{}ccccc|ccccc@{}}
        2g+\mu & -g &   &   &   & -\mu & & & &
    \\
        -g & 2g+\mu & -g &  &  &  & -\mu & & &
    \\
         & -g & 2g+\mu & -g &  &  &  & -\mu & &
    \\
          &    &        & \ddots & &   &    &        & \ddots &
    \\
          &  &  & -g & 2g+\mu & & & & & -\mu
    \\
    \hline
        -\mu & & & & & 2G+\mu & -G &   &   &    
    \\
         & -\mu & &  & & -G & 2G+\mu & -G &  &   
    \\
          &  &   -\mu  & & & & -G & 2G+\mu & -G & 
    \\
          &    &        & \ddots & &   &    &        & \ddots &
    \\
          & & & &  -\mu & &  &  & -G & 2G+\mu 
    \\
    \end{array}
    \right].
\end{align*}
\cref{fig:result-analogy_result} (b) shows the distribution of eigenvalues of $\mathcal T_J$ with $J=80$. The result indicates that the eigenvalues almost lie on the continuous line of the Bloch band $\cup_{\beta\in[-\pi,\pi]}\,\sigma(\mathcal T(\beta))$. 

Let us consider a perturbation of the matrix $\mathcal T_J$ and define
\begin{align*}
    \mathcal T_J^\varepsilon(\theta) =
    \mathcal T_J + 
    \begin{bmatrix}
        O & \varepsilon \mathrm e^{\mathrm i\theta} I 
        \\
        O & O
    \end{bmatrix}
\end{align*}
for $\varepsilon>0$ and $\theta\in \mathbb R$. We numerically calculated some eigenvalues of $\mathcal T_J^\varepsilon(\theta)$ with $\varepsilon=8\times 10^{-3}$ for varying $\theta$ and plotted their trajectory in \cref{fig:result-analogy_result} (c). We observe that the trajectory around the point $1-\mathrm i$ is formed by two distinct eigenvalues, implying the existence of an exceptional point of the family $\{ \mathcal T_J^\varepsilon(\theta) \}_{\varepsilon,\theta}$ inside the circle. These results suggest a possible link between the defective degeneracy of Bloch bands in purely periodic systems and the existence of exceptional points in finite-length systems with open boundaries.


\end{document}